\documentclass{article}

\usepackage{amsfonts}
\usepackage{amsmath}
\usepackage{amssymb}
\usepackage{graphics}
\usepackage{graphicx}
\usepackage{color}
\usepackage{float}
\usepackage{epsfig}
\usepackage{geometry}
\usepackage{indentfirst,latexsym,bm}
\usepackage{lscape}
\usepackage{tabularx}
\usepackage{enumerate}
\usepackage{longtable}
\usepackage{lscape}
\usepackage{mathrsfs}
\usepackage{bbm}

\numberwithin{equation}{section}

\newtheorem{theorem}{Theorem}[section]

\newtheorem{assumption}[theorem]{Assumption}

\newtheorem{lemma}[theorem]{Lemma}

\newtheorem{remark}[theorem]{Remark}

\newenvironment{proof}[1][Proof]{\noindent\textbf{#1.}\:}{\hfill $\square$}

\def\EE{\mathbb{E}}

\def\rm{\mathbb{R}^m}

\title{The convergence rate from discrete to continuous optimal investment stopping problem}

\author{Dingqian Sun\thanks{School of Mathematical Sciences, Fudan University, Shanghai, China, 200433. This work was completed during the visit in the University of Warwick as a joint Ph.D. student under the guidance of Dr.Liang and partially supported by China Scholarship Council; National Science Foundation of China (No. 11631004); and Science and Technology Commission of Shanghai Municipality (No. 14XD1400400).  \texttt{Email:dqsun14@fudan.edu.cn} }}

\date{}

\begin{document}
\maketitle

\begin{abstract}
We study the optimal investment stopping problem in both continuous and discrete case, where the investor needs to choose the optimal trading strategy and optimal stopping time concurrently to maximize the expected utility of terminal wealth. Based on the work \cite{HLT} with an additional stochastic payoff function, we characterize the value function for the continuous problem via the theory of quadratic reflected backward stochastic differential equation (BSDE for short) with unbounded terminal condition. In regard to discrete problem, we get the discretization form composed of piecewise quadratic BSDEs recursively under Markovian framework and the assumption of bounded obstacle, and provide some useful prior estimates about the solutions with the help of auxiliary forward-backward SDE system and Mallivian calculus. Finally, we obtain the uniform convergence and relevant rate from discretely to continuously quadratic reflected BSDE, which arise from corresponding optimal investment stopping problem through above characterization.\\

\noindent\textit{2000 MR Subject Classification}: 60G40,
\and 65C30, \and 93E20.
\end{abstract}


\section{Introduction}
In this paper, we consider a small trader in an incomplete financial market who can invest in risky stocks and a riskless asset and is also granted the right to stop the whole investment during the finite trading time interval $[0,T]$ to obtain corresponding payoff. The objective of the investor is to maximize her/his exponential utility of terminal wealth, which includes both the profit or loss on investment and the final payoff, by choosing the optimal trading strategy and optimal stopping time simultaneously. For the continuous case, the investor is allowed to stop the investment, which is like exercising an American option, at any time before $T$. While for the discrete case, the invester will be restricted to given discrete exercise time, where the payoff can be regarded as a kind of Bermudan option. 

Such utility maximization problems of mixed optimal stopping/control type have been initially studied in \cite{KW1}, which involved both consumption and final wealth under continuous framework and was reduced to a family of related pure optimal stopping problems via duality theory. Similar problems also arise in situations like pricing constrained American contingent claims, see \cite{KW2} for example, where the closed-form of hedging price of an American-type barrier option under the short-selling constraint has been obtained through the solution to a variational inequality. While different from the methods applied in these results, we will proceed by means of the connection between the original utility maximization problem (with a prespecified terminal time) and the theory of quadratic BSDEs, which will be introduced in more detail hereinafter, and pay more attention to the convergence from discrete to continuous problem.

With respect to the continuous problem, if we only consider the optimal strategy on time interval $[0,\tau]$ with fixed $\tau \in[0,T]$, it will then become the usual exponential utility maximization problem which has been widely discussed before, see \cite{Hu}, \cite{HLT}, \cite{MAM1} and \cite{MAM2}. To be specific, when the terminal payoff at $\tau$ is bounded, the problem has been completely solved in \cite{Hu} with the help of quadratic BSDE with bounded terminal data. It turns out that the value function of such problem can be characterized by the solution to a particular BSDE, whose generator is of quadratic growth in $z$-variable. Related theory of quadratic BSDE can be traced back to \cite{KobyBSDE} with bounded terminal value, where the existence and uniqueness of solutions were estabilished. Then it was extended to unbounded case to obtain the existence in \cite{PB1}, and subsequently the uniqueness with convex generators in \cite{PB},  \cite{Fre} and \cite{Fre1}. Recently, \cite{HLT} generalized the previous work, the exponential utility maximization problem with bounded payoff, to the unbounded framework on the basis of above development and studied utility indifference valuation of derivatives with unbounded payoffs as application.

Inspired by the above connection, we adjust the order of optimization and decompose our problem with extra payoff function into original utility maximization framework, which then reduced to a pure optimal stopping problem, and further obtain the value function in terms of the solution to a quadratic reflected BSDE, where the generator has almost the same form as in utility maximization problem in \cite{HLT}. While the existence and uniqueness of solutions to such quadratic reflected BSDE have been developed, see \cite{Koby} for bounded terminal value and obstacle and \cite{LX}, \cite{ErhanSong} for unbounded cases, the main difficulty left is to represent the solution to reflected BSDE via the supremum of solutions to a collection of BSDEs, which have the same quadratic generator as the former, i.e., $Y_t=\sup_{\tau\in[t,T]}Y_t(\tau), t\in[0,T]$ in subsection 2.3. Since here the group of BSDEs have different time horizon $[0,\tau]$ and terminal value $g_\tau$ and thus corresponding different pairs of solutions $(Y(\tau),Z(\tau))$, we can not directly apply the optimal stopping representation of reflected BSDEs (see Proposition 2.3 in \cite{El} for reference), but need to further use the comparison theorem and uniqueness in quadratic BSDE to prove such characterization, see Theorem \ref{continuous} for more details.

Regarding the discrete problem, we need to restructure the framework and proceed under Markovian system for the sake of following convergence analysis between the two forms. We first give a practical example to illustrate how we get the Markovian structure arising from previous continuous problem. While due to the additon of stochastic factor, the generator we considered herein will be more complicated than that in previous section, i.e., $f(t,x,z)$ of quadratic growth in $z$ and satisfying locally Lipschitz condition with respect to both $x$ and $z$, which is generalized in Assumption \ref{functions}. Then when restricted the exercise time to some given discrete time points, we can deduce recursively from the comparison result of BSDEs to get the backward discretization form, which is composed of piecewise BSDEs and actually a so-called discretely reflected BSDE, see subsection 3.2 for the form and related properties.

The main result of this paper is the convergence analysis and relevant rate from discrete to continuous optimal investment stopping problem. Thanks to previous discussion and characterization, we can now transform the problem into the convergence from discretely to continuously reflected BSDE, which has been studied when the generator is uniformly Lipschitz in all the variables, see \cite{Ma} based on the Euler scheme of forward SDE and section 3 in \cite{CHA}. Whereas originating from the utility maximization problem, we are facing reflected BSDEs with generator of quadratic growth, which brings us new difficulties during estimation and thus we have to restrict ourselves to the case of bounded and Lipschitz obstacle, and also the deterministic diffusion term in forward SDE at this stage.

Firstly,  with the help of the properties of quadratic BSDE and reflected BSDE with bounded terminal value, we can inductively prove the boundness of $\hat{Y}^{\Pi}$ in discretization form and the relationship $\hat{Y}^{\Pi}\leq Y$, which makes it possible to implement the usual techniques using to deal with  BSDE of quadratic growth.

Moreover, in order to handle the additional term coming from reflection, we need further properties about $\hat{Z}^{\Pi}$ appearing in piecewise BSDEs of the discretization form. We establish the connections between our discretization form and an auxiliary forward-backward SDE system defined on each time interval $[t_{i-1},t_i]$ with different terminal functions $\{u_i^{\Pi}\}_{1\leq i \leq n}$. Recall the existing results in Markovian FBSDE system that the solution $Z$ to quadratic BSDE with bounded and Lipschitz terminal $g(X_T)$ is controlled by $C(K_g+1)$, see \cite{Richou2011}, where $K_g$ is the Lipschitz constant of $g$. And then in \cite{Richou2012}, the prior estimate on $Z$ is generalized to the superquadratic case with unbounded terminal condition and also the case with random diffusion term in forward SDE and bounded terminal condition. While unfortunately, neither of them can cover the situation in our assumptions since here the locally Lipschitz coefficient of $x$ involves $z$. However, motivated by the proof of these results, we can make use of the BMO property of $Z$ and the representation derived from Mallivian calculus to fill this gap and get the explicit bound of $Z$. Together with the uniform Lipschitz continuity of terminal functions $\{u_i^{\Pi}\}$ in auxiliary forward-backward SDE system, we can obtain the boundness of $\hat{Z}^{\Pi}$ in discretization form at last. 

Finally, we give the complete proof of the uniform convergence from discretely to continuously quadratic reflected BSDE and obtain the convergence rate as follows when the obstacle $g$ is Lipschitz (and also the double rate if $g$ in $C_b^2$):
\begin{equation*}
\max_{1\leqslant i\leqslant n}\EE\left[\sup_{t\in[t_{i-1},t_i]}|\hat{Y}^\Pi_t-Y_t|^2\right]+\EE\left[\int^T_0 |\hat{Z}^\Pi_t-Z_t|^2dt\right] \leqslant C |\Pi|^{\frac{1}{2}}
\end{equation*}
and
\begin{equation*}
\max_{1\leqslant i\leqslant n}\EE\left[\sup_{t\in[t_{i-1},t_i]}|\hat{K}^\Pi_t-K_t|\right]
\leqslant C |\Pi|^{\frac{1}{4}}.
\end{equation*}

The paper is organized as follows. We discuss the continuous optimal investment stopping problem in section 2 and give the characterization of value function in terms of the solution to quadratic reflected BSDE. In section 3, we focus on Markovian framework and put forward the assumptions based on a practical example, and further obtain the discretization form for corresponding discrete problem. Then in section 4, after providing some auxiliary results regarding the discretization form with the aid of a forward-backward SDE system, we finally provide the convergence result of the two forms and section 5 concludes the paper.

\section{Continuous optimal investment stopping problem}

We fix a finite time horizon $[0,T]$ with $T>0$. Let $B$ be a $m$-dimensional standard Brownian motion defined on a complete probability space $(\Omega,\mathcal{F},\mathbb{P})$, and $\{\mathcal{F}_t\}_{t\geqslant 0}$ be the augumented natural filtration of $B$ which satisfies the usual conditions.

Let $\mathcal{P}$ denote the progressively measurable $\sigma$-field on $[0,T]\times\Omega$.

\subsection{Formulation}
Consider a financial market consisting of one risk-free bond with interest rate zero and $d \leqslant m$ stocks. In the case $d<m$, we face an incomplete market. The price process of the $i$th stock is described as
$$ \frac{dS^{i}_{t}}{S^{i}_{t}}=b^{i}_{t}dt+\sigma^{i}_{t}dB_{t},\quad i=1,...,d,$$
where $b^{i}$ (resp. $\sigma^{i}$) is an $\mathbb{R}$-valued (resp. $\mathbb{R}^{m}$-valued) predictable bounded stochastic process. The $\mathbb{R}^{d \times m}$-valued volatility matrix has full rank, that is, ${\sigma_{t} }{\sigma^{tr}_{t}}$ is invertible $\mathbb{P}$-a.s., for any $t\in[0,T]$. Furthermore, we assume the $\mathbb{R}^{m}$-valued risk premium process defined as
$$\theta_t=\sigma^{tr}_{t}({\sigma_{t} }{\sigma^{tr}_{t}})^{-1} b_{t},\quad t\in[0,T]$$
is also bounded.
For $i=1,...,d$, let $\pi^{i}_t$ denote the amount of money invested in stock $i$ at time $t$, and then the number of shares should be $\frac{\pi^{i}_t}{S^{i}_t}$. An $\mathbb{R}^{d}$-valued predictable process $\pi=(\pi_t)_{0\leqslant t\leqslant T}$ is called a self-financing trading strategy if $\int \pi \frac{dS}{S}$ is well defined, for example, $\int^{T}_{0}\vert\pi^{tr}_{t}\sigma_t\vert^{2}dt<\infty$, $\mathbb{P}$-a.s., which means the investor trades dynamically among the risk-free bond and the risky assets with her/his initial capital and no extra investment or withdrawal during the investment.

The wealth process with initial capital $x$ and trading strategy $\pi$ satisfies the equation
\begin{equation*}
X^{\pi}_t=x+\sum^{d}_{i=1}\int^{t}_{0}\frac{\pi^{i}_{u}}{S^{i}_{u}}dS^{i}_{u}=x+\int^{t}_{0}\pi^{tr}_{u}\sigma_{u}(dB_u+\theta_u du),\quad t\in[0,T].
\end{equation*}
Suppose there is an additional adapted process $(g_t)_{0\leqslant t\leqslant T}$ defined as the payoff at each time $t$ and recall that the investor has the right to stop at any time during the trading interval $[0,T]$, which means, if the investor chooses to stop at $\tau\in[0,T]$, then the total wealth of the investor is $X_{\tau}+g_{\tau}$. Here, $g_{\tau}\geqslant 0$ means an income, otherwise it is a flow-out. The objective of the investor is to choose both the optimal stopping time and an admissible self-financing trading strategy $\pi$ to maximize the expected utility of total wealth, which is in exponential form with the parameter $\alpha>0$, i.e.
\begin{align*}
V(0,x) &=\sup_{\tau\in[0,T]} \sup_{\pi\in \mathcal{U}_{ad}[0,\tau]} \EE\left[U_\alpha(X_{\tau}+g_{\tau})\right]\\
&=\sup_{\tau\in[0,T]} \sup_{\pi\in \mathcal{U}_{ad}[0,\tau]} \EE\left[-\exp\left(-\alpha\left(x+\int^{\tau}_{0}\pi^{tr}_{u}\frac{dS_{u}}{S_{u}}+g_{\tau}\right)\right)\right].
\end{align*}
Here $V(0,x)$ is called the value function at initial time $0$ and $\mathcal{U}_{ad}[0,\tau]$ is the admissible strategy set on $[0,\tau]$, given by Definition 1 in \cite{HLT} .

More generally, we can consider this mixed optimal stopping/control problem in dynamic form
\begin{equation}\label{valuefunc}
V(t,X_t)=\sup_{\tau\in[t,T]} \sup_{\pi\in \mathcal{U}_{ad}[t,\tau]} \EE\left[-\exp\left(-\alpha\left(X_t+\int^{\tau}_{t}\pi^{tr}_{u}\frac{dS_{u}}{S_{u}}+g_{\tau}\right)\right)\Big| \mathcal{F}_t\right],
\end{equation}
for all $t\leqslant T$. Here $X_t$ is the initial wealth when we start at the initial time $t$.

\subsection{Results on Quadratic reflected BSDEs with unbounded obstacle}
We first present the existence and uniqueness results of quadratic reflected BSDEs with the terminal data and obstacle satisfying exponential integrability, which were perfectly proved in \cite{ErhanSong}, and we will use the results to further solve the optimal investment stopping problem in continuous setting to implement the utility maximization.

A reflected BSDE with generator $f$, lower obstacle $g$ and terminal condition $g_T$ (here we only consider this special case) is an equation of the form
\begin{equation}\label{RBSDE}
g_t \leqslant Y_t=g_T+\int^{T}_{t}f(s,Z_s)ds-\int^{T}_{t}Z^{tr}_sdB_s+K_T-K_t,\quad t\in[0,T],
\end{equation}
satisfying the flat-off condition:
\begin{equation}\label{flatoff}
\int^{T}_{0}(Y_t-g_t)dK_t=0.
\end{equation}
Recall the generator $f: [0,T]\times\Omega\times\rm\rightarrow\mathbb{R}$ is a $\mathcal{P}\otimes\mathcal{B}(\rm)$ measurable function and the obstacle $g$ is an $\mathbb{R}$-valued continuous adapted process. 

Let $\EE^{\lambda, \lambda^\prime}[0,T]$ denote all the $\mathbb{R}$-valued continuous adapted processes $(Y_t)_{0\leqslant t\leqslant T}$ such that $\EE[e^{\lambda Y^{-}_*}+e^{\lambda^{\prime} Y^{+}_*}]<\infty$, where $Y^{\pm}_*\triangleq\sup_{t\in[0,T]}(Y_t)^{\pm}$ and $\EE^{p}[0,T]\triangleq\EE^{p,p}[0,T]$. $\mathbb{H}^{2p}([0,T];\rm)$ denotes all $\rm$-valued predictable processes $(Z_t)_{0\leqslant t\leqslant T}$ with $\EE(\int^{T}_0 \vert Z_t \vert^2_{\rm}dt)^p<\infty$ and $\mathbb{K}^p[0,T]$ denotes all the $\mathbb{R}$-valued continuous adapted processes $(K_t)_{0\leqslant t\leqslant T}$ , which are increasing with $K_0$=0 and $\EE\vert K_T\vert^p<\infty$.

\begin{assumption}\label{g1}
The obstacle $g$ satisfies the exponential integrable condition:
$$\EE\left[e^{\lambda\alpha g^{-}_*}+e^{\lambda^{\prime}\alpha g^{+}_*}\right]<\infty,$$
for some $\lambda, \lambda^\prime>6$ with $\frac{1}{\lambda}+\frac{1}{\lambda^\prime}<\frac{1}{6}$.
\end{assumption}

\begin{assumption}\label{g2}
The obstacle $g$ satisfies the arbitrary exponential integrable condition:
$$\EE\left[e^{p\vert g_*\vert}\right]<\infty,\quad\forall p\geqslant 1.$$
\end{assumption}

\begin{theorem}\label{RBSDEsolution}
Suppose that Assumption \ref{g1} holds with parameters $\lambda$ and $\lambda^\prime$. Then, the quadratic reflected BSDE (\ref{RBSDE}) and (\ref{flatoff}) with generator
\begin{equation}\label{f}
f(t,z)=-\frac{\alpha}{2} \min_{\pi_t \in \mathcal{C}}\Bigg|\sigma^{tr}_t \pi_t-\left(\frac{1}{\alpha}\theta_t-z\right)\Bigg|^2-z^{tr}\theta_t+\frac{1}{2\alpha}\vert\theta_t\vert^2
\end{equation}
admits a unique solution $(Y,Z,K)\in \bigcap_{p\in(1,\frac{\lambda\lambda^\prime}{\lambda+\lambda^\prime})}\EE^{\lambda\alpha, \lambda^\prime\alpha}[0,T]\times \mathbb{H}^{2p}([0,T];\rm)\times \mathbb{K}^p[0,T]$. Here $\mathcal{C}$ is a closed and convex set in the definition of admissible strategy satisfying $0\in \mathcal{C}$, which the strategy can take values in.

In addition, if $g$ satisfies Assumption \ref{g2}, then the unique solution belongs to $\EE^{p}[0,T]\times \mathbb{H}^{2p}([0,T];\rm)\times \mathbb{K}^p[0,T]$ for all $p\in[1,\infty)$, i.e.,
$$\EE\left[e^{{p\gamma Y_*}}+\left(\int^{T}_0 \vert Z_s\vert^2ds\right)^p+K^p_T\right]<\infty.$$
\end{theorem}

\begin{proof}
One can easily check that $f$ with the form (\ref{f}) satisfies
\begin{equation}\label{fcondition}
-\frac{\alpha}{2}\vert z\vert^2 \leqslant f(t,z)\leqslant -z^{tr}\theta_t+\frac{1}{2\alpha}\vert \theta_t\vert^2,
	\end{equation}
and is concave in $z$, i.e., it satisfies Assumptions (H1) and (H3) in \cite{ErhanSong}. Consequently, we can get the existence and uniqueness directly from Theorem 3.2 and 4.1 there.
\end{proof}
\\

\subsection{Characterization of value funtion}
Now we can characterize the value function of the optimal problem by using the solution to the above reflected BSDE.

\begin{theorem}\label{continuous}
Suppose that $g$ satisfies Assumption \ref{g2} and let $(Y,Z,K)$ be the unique solution to quadratic reflected BSDE (\ref{RBSDE}) and (\ref{flatoff}) with generator (\ref{f}). Then, the value function (\ref{valuefunc}) of the continuous optimal investment stopping problem can be given by
\begin{equation*}
V(t,X_t)=-\exp(-\alpha(X_t+Y_t)),\quad \forall t\in[0,T].
\end{equation*}
\end{theorem}

\begin{proof}
For any fixed $t\in[0,T]$ and $\tau\in[t,T]$, we first solve the optimal control problem for the time interval $[t,\tau]$ by considering the following quadratic BSDE
\begin{equation}\label{BSDE}
Y_t(\tau)=g_\tau+\int^{\tau}_{t}f(s,Z_s(\tau))ds-\int^{\tau}_{t}Z^{tr}_s(\tau)dB_s,
\end{equation}
where the generator $f$ has the same form as in reflected BSDE, i.e., satisfies (\ref{f}). For convenience, we will note the above equation as BSDE $(f, g_\tau)$ thereafter. Additionally, we denote the solution to this BSDE as $(Y_{.}(\tau), Z_{.}(\tau))$ in order to emphasize its dependence on the terminal time $\tau$ and corresponding terminal value $g_\tau$. Then we can represent the latter part of the value function by dynamic programming principle as follows,
\begin{equation}\label{vf1}
\sup_{\pi\in \mathcal{U}_{ad}[t,\tau]} \EE\left[-\exp\left(-\alpha\left(X_t+\int^{\tau}_{t}\pi^{tr}_{u}\frac{dS_{u}}{S_{u}}+g_{\tau}\right)\right)\Big| \mathcal{F}_t\right]=-\exp(-\alpha(X_t+Y_t(\tau))),
\end{equation}
based on the existing result in \cite{HLT}, see Theorem 6. In turn, the original mixed optimal stopping/control problem becomes
\begin{equation}\label{vf2}
V(t,X_t)=\sup_{\tau\in[t,T]}[-\exp(-\alpha(X_t+Y_t(\tau)))]=-\exp\left(-\alpha\left(X_t+\sup_{\tau\in[t,T]}Y_t(\tau)\right)\right),
\end{equation}
and we only need to show that $Y_t=\sup_{\tau\in[t,T]}[Y_t(\tau)]$ for any $t\in[0,T]$.

First, for any $0\leqslant t\leqslant \tau\leqslant T$, let
$$\mathcal{Y}_t\triangleq Y_\tau+\int^{\tau}_{t}f(s,Z_s)ds-\int^{\tau}_{t}Z^{tr}_sdB_s $$
and we have $Y_t=\mathcal{Y}_t+K_\tau-K_t$.
Recalling that $(Y(\tau), Z(\tau))$ satisfies (\ref{BSDE}) on $[t,\tau]$ with the same generator as $(\mathcal{Y}, Z)$ and their terminal values satisfy $Y_\tau\geqslant g_\tau$, we can then deduce that $\mathcal{Y}_t\geqslant Y_t(\tau)$ via the comparison theorem of quadratic BSDE, see Theorem 5 in \cite{PB}, where the proof and result can be easily adapted to the case of concave generator with quadratic growth from below. Moreover, since $K$ is an increasing process, we have $K_\tau-K_t\geqslant 0$ and thus $Y_t=\mathcal{Y}_t+K_\tau-K_t\geqslant Y_t(\tau)$ for any $\tau\in[t,T]$, which gives rise to $Y_t\geqslant \sup_{\tau\in[t,T]}[Y_t(\tau)]$.

The idea of the following proof comes from representation of the solution to reflected BSDE, which is corresponded to an optimal stopping problem, see \cite{El}. For any $t\in[0,T]$, define $D_t\triangleq \inf \{ s\in[t,T]: Y_s=g_s\}$ and since $Y_T=g_T$, we can obtain $t \leqslant D_t\leqslant T$. Considering the reflected BSDE on interval $[t,D_t]$,
$$Y_t=Y_{D_t}+\int^{D_t}_{t}f(s,Z_s)ds-\int^{D_t}_{t}Z^{tr}_sdB_s+K_{D_t}-K_t,$$
by the continuity of $K$ and the flat-off condition (\ref{flatoff}), we have $K_{D_t}=K_t$ (which means $K_s\equiv K_t$ for any $s\in[t,D_t]$) and then $(Y.,Z.)$ becomes the solution to BSDE $(f, Y_{D_t})$ on $[t,D_t]$. In the meanwhile, note that $(Y.(D_t), Z.(D_t))$ is the solution to BSDE $(f, g_{D_t})$ on $[t,D_t]$ and the definition of $D_t$ futher yields $Y_{D_t}=g_{D_t}$. Thus by the uniqueness of solution to quadratic BSDE, see \cite{Fre}, we have $Y.=Y.(D_t)$ on $[t,D_t]$, and specifically $Y_t=Y_t(D_t)$, which completes the proof.
\end{proof}

\begin{remark}
We need to note here that for convenience, what we discussed in this paper is quadratic reflected BSDE with lower obstacle, whose solution we have proved in the above theorem can be characterized by the supremum of the solutions to a collection of BSDEs with the same generator. Therefore, we require consistency of the supremum whether it is taking inside or outside the exponential in the expression of value function (\ref{vf2}). To this end, when quoting the result in \cite{HLT}, we have to change the sign of $Y$appearing in the value function as (\ref{vf1}) and then the corresponding generator of BSDE. Actually, denoting the generator there as $F$, one can readily check that they satisfy $f(t,z)=-F(t,-z)$ and that is why we are considering concave generator in this section.
\end{remark}

\section{Discrete optimal investment stopping problem}

From this section, we will concentrate on Markovian framework, that is, the following decoupled forward-backward SDE with reflection
\begin{align}\label{M-FBSDE}
\begin{split}
& X_t=x+\int^t_0 b(s,X_s)ds+\int^t_0 \sigma(s)dB_s,\\
&Y_t=g(X_T)+\int^{T}_{t}f(s,X_s,Z_s)ds-\int^{T}_{t}Z^{tr}_sdB_s+K_T-K_t,\,t\in[0,T],\\
&Y_t\geqslant g(X_t)\quad\text{and}\quad \int^{T}_{0}(Y_t-g(X_t))dK_t=0.
\end{split}
\end{align}

For the functions that appear in the above system, we have following general assumptions.
\begin{assumption}\label{functions}
$b,\sigma,g$ and $f$ are deterministic functions that satisfy:\\
\noindent$\mathbf{(a)} \, b: [0,T]\times\mathbb{R}\rightarrow\mathbb{R}$ and $\sigma: [0,T]\rightarrow\rm$ are continuous functions and there exists constants $M_b, K_b$ and $M_\sigma$ such that $\forall t\in[0,T], \forall x,x^\prime \in \mathbb{R}$,
\begin{align*}
\vert b(t,x)\vert &\leqslant M_b(1+\vert x\vert),\\
\vert b(t,x)-b(t,x^{\prime})\vert &\leqslant K_b\vert x-x^{\prime}\vert,\\
\vert \sigma(t)\vert &\leqslant M_\sigma.
\end{align*}
\noindent$\mathbf{(b)} \, f:  [0,T]\times\mathbb{R}\times\rm\rightarrow\mathbb{R}$ and $g: \mathbb{R}\rightarrow\mathbb{R}$ are continuous functions and there exists constants $M_f, K_x, K_z, K_g$ and $M_g$ such that  $\forall t\in[0,T], \forall x,x^\prime \in \mathbb{R}$ and $\forall z,z^\prime \in \rm$,
\begin{align*}
\vert f(t,x,z)\vert &\leqslant M_f+\frac{\alpha}{2}\vert z\vert^2,\\
\vert f(t,x,z)-f(t,x^{\prime},z)\vert & \leqslant K_x(1+\vert z\vert)\vert x-x^{\prime}\vert,\\
\vert f(t,x,z)-f(t,x,z^{\prime})\vert & \leqslant K_z(1+\vert z\vert+\vert z^{\prime}\vert)\vert z-z^{\prime}\vert,\\
\vert g(x)-g(x^{\prime})\vert &\leqslant K_g\vert x-x^{\prime}\vert,\\
\vert g(x)\vert &\leqslant M_g.
\end{align*}
\end{assumption}

Let $\mathbb{S}^\infty[0,T]$ denote the set of $\mathbb{R}$-valued progressively measurable bounded processes and $\mathbb{S}^p[0,T]$ denote the space of all $\mathbb{R}$-valued adapted processes $(Y_t)_{t\in[0,T]}$ such that $\EE[\sup_{0\leqslant t\leqslant T} \vert Y_t\vert^p] <\infty$. Then under the above assumptions, we know the decoupled system (\ref{M-FBSDE}) with bounded terminal condition and bounded obstacle has a unique solution $(X,Y,Z,K)\in \mathbb{S}^2[0,T]\times\mathbb{S}^\infty[0,T]\times\mathbb{H}^2([0,T];\rm)\times\mathbb{K}^2[0,T]$. For more details of this result, we refer to \cite{Koby}.

\subsection{A special case as connection}
We will see from a special case with subspace portfolio constraints in this subsection that how we can get the above Markovian structure from the previous general problem. Here for simplicity, we consider a market with a single stock whose coefficients depend on a single stochastic factor driven by a 2-dim Brownian motion, that is, $m=2$, $d=1$ and 
\begin{equation*}
\frac{dS_t}{S_t}=b(t,V_t)dt+\sigma(t,V_t)dB_{1,t},
\end{equation*}
\begin{equation}\label{SF}
dV_t=\eta(V_t)dt+(\kappa_1,\kappa_2)
\begin{pmatrix} dB_{1,t}\\dB_{2,t} \end{pmatrix},
\end{equation}
where $\kappa_1, \kappa_2$ are two positive constants satisfying $|\kappa_1|^2+|\kappa_2|^2=1$. We assume that $b, \sigma$ and $\eta$ are uniformly bounded and Lipschitz with respect to $x$, and furthermore, $\sigma \geqslant\delta$ for some $\delta>0$. Then the wealth process is
\begin{equation*}
dX_t=\pi_t \frac{dS_t}{S_t}=\pi_t \left[b(t,V_t)dt+\sigma(t,V_t)dB_{1,t} \right].
\end{equation*}
Setting $\mathcal{C}=\mathbb{R}$ and $\theta(t,V_t)\triangleq \frac{b(t,V_t)}{\sigma(t,V_t)}$, we know from the above assumptions that $\theta$ is also both bounded and Lipschitz. Supposing further that the payoff has the form as a function of stochastic factor $V$, that is, $g(V_\cdot)$, the reflected BSDE (\ref{RBSDE}) will then become
\begin{equation}\label{Eg1}
g(V_t) \leqslant Y_t=g(V_T)+\int^{T}_{t}f(s,Z_s)ds-\int^{T}_{t}Z^{tr}_sdB_s+K_T-K_t,\quad t\in[0,T],
\end{equation}
where the generator in (\ref{f}) reduces to
\begin{equation*}
f(t,z)=f(t,(z_1,z_2))=-\frac{\alpha}{2}|z_2|^2-z_1\theta(t,V_t)+\frac{1}{2\alpha}|\theta(t,V_t)|^2.
\end{equation*}

If we regard the equation of stochastic factor (\ref{SF}) as the forward SDE and let $\bar{f}(t,x,z)\triangleq -\frac{\alpha}{2}|z_2|^2-z_1\theta(t,x)+\frac{1}{2\alpha}|\theta(t,x)|^2$, then $\bar{f}$ is now a deterministic function and (\ref{Eg1}) becomes
\begin{equation}
g(V_t) \leqslant Y_t=g(V_T)+\int^{T}_{t}\bar{f}(s,V_s,Z_s)ds-\int^{T}_{t}Z^{tr}_sdB_s+K_T-K_t,\quad t\in[0,T].
\end{equation}
Combined with (\ref{SF}), they constitute a Markovian system as (\ref{M-FBSDE}) and one can easily check that $\bar{f}$ satisfies the above Assumption. 

In order to avoid confusion about the notations, we will still use $b$ and $\sigma$ to denote the coefficients of forward SDE and $(X,Y,Z,K)$ the solution of forward-backward SDE with reflection in the following discussion, and consider the discrete problem and subsequent convergence under the generalized Assumption \ref{functions}.

\subsection{Discretization form}

We  continue to consider the optimal investment stopping problem in a discrete setting, which means the investor is only allowed to stop the investment process at given discrete time points $\Pi\triangleq\{t_i,\, i=0,1,\ldots,n\,\vert \,0=t_0<t_1<t_2<\cdots<t_n=T\}$. Denote $D[t,T]\triangleq [t,T]\cap\Pi$ and $\Delta t_i=t_i-t_{i-1}$ for $i=1,\ldots,n$, and let $|\Pi|\triangleq \max_{1\leqslant i\leqslant n}\Delta t_i$. The corresponding value function for discrete problem becomes
\begin{align*}
&\sup_{\tau\in D[t,T]} \sup_{\pi\in \mathcal{U}_{ad}[t,\tau]} \EE\left[-\exp\left(-\alpha\left(X_t+\int^{\tau}_{t}\pi^{tr}_{u}\frac{dS_{u}}{S_{u}}+g(X_{\tau})\right)\right)\Big| \mathcal{F}_t\right] \\
&= \sup_{\tau\in D[t,T]}[-\exp(-\alpha(X_t+Y_t(\tau)))]\\
&=-\exp\left[-\alpha\left(X_t+\max_{\tau\in D[t,T]}Y_t(\tau)\right)\right],
\end{align*}
where $Y_\cdot (\tau)$ satisfies the BSDE
\begin{equation*}
Y_t(\tau)=g(X_\tau)+\int^{\tau}_{t}f(s,X_s,Z_s(\tau))ds-\int^{\tau}_{t}Z^{tr}_s(\tau)dB_s,\quad t\in[0,\tau].
\end{equation*}

Define $\hat{Y}^{\Pi}_t\triangleq \max_{\tau\in D[t,T]}Y_t(\tau)=\max_{\tau\in D[t,T]}\EE[g(X_{\tau})+\int^{\tau}_t f(s,X_s, Z_s(\tau))ds\vert \mathcal{F}_t]$ for any $t\in [0,T]$. Then the value function in discrete case turns out to be $V^{\Pi}(t,X_t)=-\exp[-\alpha(X_t+\hat{Y}^{\Pi}_t)]$, which indicates that we only need to focus on the difference between $\hat{Y}^{\Pi}$ and $Y$. Thanks to the comparison result of quadratic BSDEs, we can characterize $\hat{Y}^{\Pi}$ inductively and it is actually a so-called discretely reflected BSDE, which means that reflection only operates at specific time points ${\Pi}$. We will depict the processes $\bar{Y}^{\Pi}$ and $(\hat{Y}^{\Pi},\hat{Z}^{\Pi},\hat{K}^{\Pi})$ recursively as follows and in order to simplify the notation, we will proceed with the case $m=1$, while one can easily generalize the results to $m$-dimension:\\

\noindent$\bullet$ $\hat{Y}^{\Pi}_{t_n}=\bar{Y}^{\Pi}_{t_n}=g(X_T);$\\
\noindent$\bullet$ For $i=n, n-1,\cdots,1$ and $t\in[t_{i-1},t_i)$, $(\bar{Y}^{\Pi},\hat{Z}^{\Pi})$ is the solution to quadratic BSDE:
\begin{equation}\label{Ybar}
\bar{Y}^{\Pi}_t=\hat{Y}^{\Pi}_{t_i}+\int^{t_i}_t f(r,X_r,\hat{Z}^{\Pi}_r)dr-\int^{t_i}_t \hat{Z}^{\Pi}_r dB_r;
\end{equation}
\noindent$\bullet$ For $i=n, n-1,\cdots,1$, define $\hat{Y}^{\Pi}_{t}=\bar{Y}^{\Pi}_t$ for any $t\in(t_{i-1},t_i)$ and $\hat{Y}^{\Pi}_{t_{i-1}}=\bar{Y}^{\Pi}_{t_{i-1}}\vee g(X_{t_{i-1}})$;\\
\noindent$\bullet$ Let $\hat{K}^{\Pi}_0\triangleq 0$ and for $i=1,2,\cdots,n$, $t\in(t_{i-1},t_i]$, and define $\hat{K}^{\Pi}_t\equiv \hat{K}^{\Pi}_{t_i}\triangleq \sum^i_{j=1}(\hat{Y}^{\Pi}_{t_{j-1}}-\bar{Y}^{\Pi}_{t_{j-1}})$.\\

Since $\hat{K}^{\Pi}_{t_i}\in\mathcal{F}_{t_{i-1}}$ for any $1\leqslant i\leqslant n$, we know that $\hat{K}^{\Pi}$ is $\{\mathcal{F}_{t}\}$-predictable. In addition, we can deduce from definition that $\hat{K}^{\Pi}_{t_i}-\hat{K}^{\Pi}_{t_{i-1}}=\hat{Y}^{\Pi}_{t_{i-1}}-\bar{Y}^{\Pi}_{t_{i-1}}$, which leads to
\begin{equation}\label{Yhat}
\hat{Y}^{\Pi}_{t_{i-1}}=\hat{Y}^{\Pi}_{t_i}+\int^{t_i}_{t_{i-1}} f(r,X_r,\hat{Z}^{\Pi}_r)dr-\int^{t_i}_{t_{i-1}} \hat{Z}^{\Pi}_r dB_r+\hat{K}^{\Pi}_{t_i}-\hat{K}^{\Pi}_{t_{i-1}},
\end{equation}
and that is why it is called discretely reflected BSDE.

\begin{lemma}\label{bddY}
Let Assumption \ref{functions} hold. Then, we have $(\romannumeral1)$ both $\bar{Y}^{\Pi}$ and $\hat{Y}^{\Pi}$ are bounded by $M_g+M_f T$, uniformly in $\Pi$; $({\romannumeral2})$
$\bar{Y}^{\Pi}_{t}\leqslant \hat{Y}^{\Pi}_{t}\leqslant Y_t$ for all $t\in[0,T]$.
\end{lemma}
\begin{proof}
$(\romannumeral1)$ For the first claim, since $\Vert \bar{Y}^{\Pi}\Vert_\infty\leqslant\Vert \hat{Y}^{\Pi}_{t_n}\Vert_\infty+M_f \Delta t_n=\Vert g(X_{T})\Vert_\infty+M_f\Delta t_n\leqslant M_g+M_f\Delta t_n\leqslant M_g+M_f T$ on $[t_{n-1},t_n)$ by Corollary 2.2 in \cite{KobyBSDE} and $\vert \hat{Y}^{\Pi}_{t_{n-1}}\vert\leqslant\vert\bar{Y}^{\Pi}_{t_{n-1}}\vert\vee \vert g(X_{t_{n-1}})\vert\leqslant M_g+M_f\Delta t_n$, the conclusion holds for the first interval $[t_{n-1},t_n)$ and also for $t=t_n$.

Then for the next interval $[t_{n-2},t_{n-1})$, using the Corollary again we have $\Vert\bar{Y}^{\Pi}\Vert_\infty\leqslant\Vert \hat{Y}^{\Pi}_{t_{n-1}}\Vert_\infty+M_f\Delta t_{n-1}\leqslant M_g+M_f(\Delta t_{n-1}+\Delta t_{n})\leqslant M_g+M_f T$ on $[t_{n-2},t_{n-1})$, and similarly $|\hat{Y}^{\Pi}_{t_{n-2}}|\leqslant|\bar{Y}^{\Pi}_{t_{n-2}}|\vee |g(X_{t_{n-2}})|\leqslant M_g+M_f(\Delta t_{n-1}+\Delta t_{n})$.

By analogy, we can finally obtain $\Vert\bar{Y}^{\Pi}\Vert_\infty\vee\Vert\hat{Y}^{\Pi}\Vert_\infty\leqslant M_g +M_f(\sum^n_{j=i} \Delta t_{j})\leqslant M_g+M_f T$ on $[t_{i-1},t_i)$ for any $i=1,\ldots,n$, i.e., $\Vert\bar{Y}^{\Pi}\Vert_\infty$ and $\Vert\hat{Y}^{\Pi}\Vert_\infty$ are bounded by $M_g+M_f T$ on the whole interval $[0,T]$, and the bound is obviously independent of $\Pi$.

$({\romannumeral2})$ Observing that $\bar{Y}^{\Pi}$ and $\hat{Y}^{\Pi}$ may not be equal only on $\Pi$, one can easily get the first inequality by definition. As for the second one, we first have $ \hat{Y}^{\Pi}_{t_n}=Y_{t_n}=g(X_T)$. Assume $\hat{Y}^{\Pi}_{t_i}\leqslant Y_{t_i}$ holds. Then, similarly as the arguments in the proof of Theorem \ref{continuous}, comparing (\ref{Ybar}) and 
\begin{equation*}
Y_t=Y_{t_i}+\int^{{t_i}}_{t}f(s,X_s,Z_s)ds-\int^{{t_i}}_{t}Z_sdB_s+K_{t_i}-K_t,
\end{equation*}
the comparison result of quadratic BSDE with bounded terminals and the fact $K$ is increasing further yield that $\bar{Y}^{\Pi}_{t}\leqslant Y_t$ for any $t\in[ t_{i-1},t_i)$. Moreover, since $g(X_t)$ is the lower obstacle of $Y_t$, we have $\hat{Y}^{\Pi}_{t_{i-1}}\leqslant \bar{Y}^{\Pi}_{t_{i-1}}\vee g(X_{t_{i-1}}) \leqslant Y_{t_{i-1}}$, and thus $\hat{Y}^{\Pi}_{t}\leqslant Y_t$ for $t\in[ t_{i-1},t_i)$. Then by induction, we can conclude the second inequality.

\end{proof}

\section{Convergence analysis}
In consideration of the connections we have built respectively for the continuous and discrete optimal investment stopping problem in previous sections, we may now lay emphasis on the convergence from discretely to continuously quadratic reflected BSDE.

\subsection{Auxiliary results}
We will introduce a forward-backward SDE system on each interval $[t_{i-1},t_i]$ instead of analyzing the discretization form directly. Define $u^{\Pi}_n(x)=g(x)=\tilde{Y}_T$ and for $i=n, n-1,\cdots, 1$, let $(\tilde{Y},\tilde{Z})$ be the solution to the BSDE defined piecewise by
\begin{equation}\label{u-FBSDE}
\tilde{Y}_t=u^{\Pi}_i(X_{t_i}({t_{i-1}},x))+\int^{t_i}_t f(r,X_r({t_{i-1}},x),\tilde{Z}_r)dr-\int^{t_i}_t \tilde{Z}_r dB_r,\quad t\in[t_{i-1},t_i),
\end{equation}
where $X_\cdot ({t_{i-1}},x)$ represents the solution to forward SDE in (\ref{M-FBSDE}) starting from $({t_{i-1}},x)$. Let $u^{\Pi}_{i-1}(x)=\tilde{Y}_{t_{i-1}}(x)\vee g(x)$ and notice that here we write as the form $\tilde{Y}_{t_{i-1}}(x)$ in order to show the dependence of $\tilde{Y}$ on the initial value $x$ of the SDE.

\begin{lemma}\label{u_i}
By the definition of the collection of functions $\{u^{\Pi}_i\}_{1\leqslant i\leqslant n}$, we have $\hat{Y}^{\Pi}_{t_{i}}=u^{\Pi}_i(X_{t_i})$. Here $X_{t_i}\triangleq X_{t_i}(0,x)$.
\end{lemma}

\begin{proof}
We will prove this lemma by induction. Firstly, for $i=n$, $\hat{Y}^{\Pi}_{t_{n}}=g(X_{T})=u^{\Pi}_n(X_{t_n})$. If we assume that the result holds for $i$, i.e., $\hat{Y}^{\Pi}_{t_{i}}=u^{\Pi}_{i}(X_{t_{i}})$, then when it comes to $i-1$, we have $\hat{Y}^{\Pi}_{t_{i-1}}=\bar{Y}^{\Pi}_{t_{i-1}}\vee g(X_{t_{i-1}})$ and $u^{\Pi}_{i-1}(X_{t_{i-1}})=\tilde{Y}_{t_{i-1}}(X_{t_{i-1}})\vee g(X_{t_{i-1}})$ separately. 

Comparing (\ref{Ybar}) and (\ref{u-FBSDE}) when $x=X_{t_{i-1}}$ and noticing that $X_{r}\triangleq X_{r}(0,x)=X_r(t_{i-1},X_{t_{i-1}})$ for any $r\in[t_{i-1},t_i]$, we know that the two BSDEs have the same generator. Especially, we have $X_{t_i}({t_{i-1}},X_{t_{i-1}})=X_{t_i}$ and then $u^{\Pi}_i(X_{t_i}({t_{i-1}},X_{t_{i-1}}))=u^{\Pi}_i(X_{t_i})=\hat{Y}^{\Pi}_{t_{i}}$ by assumption, which means the BSDEs have the same terminal value as well, which is bounded as proved. Then by the uniqueness of quadratic BSDE with bounded terminal condition, see \cite{KobyBSDE}, we can conclude that $\bar{Y}^{\Pi}_t=\tilde{Y}_t(X_{t_{i-1}})$ on $[{t_{i-1}},t_i)$, and specifically, $\bar{Y}^{\Pi}_{t_{i-1}}=\tilde{Y}_{t_{i-1}}(X_{t_{i-1}})$. Consequently, we obtain $\hat{Y}^{\Pi}_{t_{i-1}}=u^{\Pi}_{i-1}(X_{t_{i-1}})$ by taking maximum with $g(X_{t_{i-1}})$ on both sides, which completes the proof.
\end{proof}\\

Let us introduce the following more general forward-backward SDE on $[0,T]$ for later use,

\begin{align}\label{M-BSDE}
\begin{split}
& X_t=x+\int^t_0 b(s,X_s)ds+\int^t_0 \sigma(s)dB_s,\\
&Y_t=g(X_T)+\int^{T}_{t}f(s,X_s,Z_s)ds-\int^{T}_{t}Z_sdB_s,
\end{split}
\end{align}
and give a crucial estimate of $Z$ in next lemma.

\begin{lemma}\label{bddZ}
Suppose Assumption \ref{functions} holds. Then, there exists a version of $Z$ such that $\forall t\in[0,T]$,
$$|Z_t|\leqslant \exp(KT)[M_\sigma \exp(2K_b T) K_g+1],$$
where $K\triangleq M_\sigma K_x\exp(2K_b T).$
\end{lemma}

The proof is given in the appendix. Since the locally Lipschitz condition of $f$ with respect to $x$ in Assumption \ref{functions} involves $z$, we can not use the existing result, like in \cite{Richou2011} and \cite{Richou2012}, directly. Fortunately, we know that with bounded terminal value, the martingale $Z\ast B$ belongs to the space of BMO martingales, which can essentially help us to prove the boundness of $Z$ and further give the explicit form. \\

Next, let us give a useful lemma called discrete backward Gronwall Inequality, which will play an important role in the following content.
\begin{lemma}\label{Gronwall}
Let $\Pi$ and $\Delta t_i$ define as above. Suppose that $\{a_i,b_i\}^n_{i=1}$ satisfy $a_i\geqslant 0, b_i\geqslant 0$, and $a_{i-1}\leqslant e^{C \Delta t_i}a_i+b_i$ for $i=2,\cdots,n$, then
\begin{equation*}
\max_{1\leqslant i\leqslant n} a_i \leqslant e^{CT}\left[a_n+\sum^n_{i=1}b_i\right].
\end{equation*}
\end{lemma}

\begin{proof}
By backward induction, we have
\begin{equation*}
\begin{split}
&a_{n-1} \leqslant e^{C \Delta t_n}a_n+b_n\leqslant e^{C \Delta t_n}[a_n+b_n],\\
&a_{n-2} \leqslant e^{C \Delta t_{n-1}}a_{n-1}+b_{n-1}
\leqslant e^{C (\Delta t_{n-1}+\Delta t_n)}[a_n+b_n+b_{n-1}],\\
&\ldots \ldots
\end{split}
\end{equation*}
thus one can easily get that
$$a_i \leqslant  e^{C (T- t_i)}\left[a_n+\sum^n_{j=i+1}b_j\right],\quad i=1,\cdots,n,$$
which completes the proof.
\end{proof}\\

Now we can consider further property of the collection of functions $\{u^{\Pi}_i\}_{1\leqslant i\leqslant n}$.

\begin{lemma}\label{uLip}
Suppose Assumption \ref{functions} holds. Then $u^{\Pi}_i$ is bounded and Lipschitz continuous, uniformly in $\Pi$ and $i$.
\end{lemma}

\begin{proof}
The first assertion in regard to the boundness actually can be proved following the same procedure as in Lemma \ref{bddY}. Now let us prove by induction that each $u^{\Pi}_i$ is Lipschitz continuous. Clearly, $u^{\Pi}_n=g$ is Lipschitz by assumption and the Lipschitz constant is $L_n\triangleq K_g$. Assuming that $u^{\Pi}_i$ is Lipschitz with constant $L_i$, then we need to show the result for $u^{\Pi}_{i-1}$. 

For any $x_1, x_2\in \mathbb{R}$,  $t\in[t_{i-1},t_i)$, denote $(\tilde{Y}^j,\tilde{Z}^j)$ as the solution to (\ref{u-FBSDE}) with initial value $x_j, j=1,2$. Regarding (\ref{u-FBSDE}) as the system (\ref{M-BSDE}) on $[t_{i-1},t_i]$ with the terminal value function $u^{\Pi}_i$ and by the Lipschitz and boundness assumption, we can get from Lemma \ref{bddZ} that $|\tilde{Z}^j_t|\leqslant C(1+L_i)$ on $[t_{i-1},t_i]$, where C denotes the constant which may depend on $T$ and all the constants appearing in the Assumption except $K_g$ and may vary from line to line. 
Consider the difference between the two solutions
\begin{equation}\label{Q^i}
\begin{split}
\tilde{Y}^1_t-\tilde{Y}^2_t=&u^{\Pi}_i(X_{t_i}(t_{i-1},x_1))-u^{\Pi}_i(X_{t_i}(t_{i-1},x_2))
-\int^{t_i}_t (\tilde{Z}^1_r-\tilde{Z}^2_r)dB_r\\
&+\int^{t_i}_t \left[ f(r,X_r(t_{i-1},x_1),\tilde{Z}^1_r)-f(r,X_r(t_{i-1},x_2),\tilde{Z}^2_r) \right]dr,\quad t\in[t_{i-1},t_i),
\end{split}
\end{equation}
and define
$$V^i_t=\left\{
\begin{array}{lcl}
\frac{f(t,X_t(t_{i-1},x_1),\tilde{Z}^1_t)-f(t,X_t(t_{i-1},x_1),\tilde{Z}^2_t)}{\tilde{Z}^1_t-\tilde{Z}^2_t} \mathbbm{1}_{\{\tilde{Z}^1_t\neq\tilde{Z}^2_t\}},&      &t\in[t_{i-1},t_i];\\
V^i_{t_{i-1}}, &      &t\in[0, t_{i-1}).
\end{array} \right.$$
Noting that $|V^i_t|\leqslant K_z(1+|\tilde{Z}^1_t|+|\tilde{Z}^2_t|)\leqslant C(1+L_i)$ for all  $t\in[0,t_i]$, we can rewrite (\ref{Q^i}) as
\begin{equation*}
\begin{split}
\tilde{Y}^1_t-\tilde{Y}^2_t=&u^{\Pi}_i(X_{t_i}(t_{i-1},x_1))-u^{\Pi}_i(X_{t_i}(t_{i-1},x_2))
-\int^{t_i}_t (\tilde{Z}^1_r-\tilde{Z}^2_r)dB_r
+\int^{t_i}_t (\tilde{Z}^1_r-\tilde{Z}^2_r)V^i_r dr\\
&+\int^{t_i}_t \left[f(r,X_r(t_{i-1},x_1),\tilde{Z}^2_r)-f(r,X_r(t_{i-1},x_2),\tilde{Z}^2_r)\right]dr\\
= &u^{\Pi}_i(X_{t_i}(t_{i-1},x_1))-u^{\Pi}_i(X_{t_i}(t_{i-1},x_2))
-\int^{t_i}_t (\tilde{Z}^1_r-\tilde{Z}^2_r)dB^{\mathbb{Q}^i}_r\\
&+\int^{t_i}_t \left[f(r,X_r(t_{i-1},x_1),\tilde{Z}^2_r)-f(r,X_r(t_{i-1},x_2),\tilde{Z}^2_r)\right]dr.
\end{split}
\end{equation*}
Here for the second equality, since $V^i_t$ is bounded on $[0,t_i]$, we can define an equivalent martingale measure $\mathbb{Q}^i$ on $\mathcal{F}_{t_i}$ by $\frac{d\mathbb{Q}^i}{d\mathbb{P}}=\mathcal{E}_{t_i}(\int^{\cdot}_0V^i_rdB_r)$, then we have that $B^{\mathbb{Q}^i}_t\triangleq B_t-\int^t_0 V^i_rdr$ is a standard Brownian motion under $\mathbb{Q}^i$. Since $\tilde{Z}^j$ is bounded on $[t_{i-1},t_i]$, we can obtain 
\begin{equation*}
\begin{split}
|\tilde{Y}^1_t-\tilde{Y}^2_t|  \leqslant &\EE^{\mathbb{Q}^i}_t \left|u^{\Pi}_i(X_{t_i}(t_{i-1},x_1))-u^{\Pi}_i(X_{t_i}(t_{i-1},x_2))\right|\\
&+\EE^{\mathbb{Q}^i}_t \left[  \int^{t_i}_t \left|f(r,X_r(t_{i-1},x_1),\tilde{Z}^2_r)-f(r,X_r(t_{i-1},x_2),\tilde{Z}^2_r)\right|dr\right] \\
\leqslant & L_i\EE^{\mathbb{Q}^i}_t|X_{t_i}(t_{i-1},x_1)-X_{t_i}(t_{i-1},x_2)|\\
&+\EE^{\mathbb{Q}^i}_t \left[  \int^{t_i}_t K_x(1+|\tilde{Z}^2_r|)|X_r(t_{i-1},x_1)-X_r(t_{i-1},x_2)|dr\right] \\
\leqslant & [L_i + C(1+L_i) \Delta t_i ]e^{K_b\Delta t_i }|x_1-x_2|,\quad t\in[t_{i-1},t_i),
\end{split}
\end{equation*}
where the last inequality above comes from standard estimate of forward SDE with deterministic diffusion term. Thus, we have $|\tilde{Y}^1_{t_{i-1}}-\tilde{Y}^2_{t_{i-1}}|\leqslant( L_i e^{K_1 \Delta t_i }+K_2 \Delta t_i)|x_1-x_2|$ by letting $K_1=C+K_b$ and $K_2=Ce^{K_b T}$. According to the definition of $u^\Pi_{i-1}$ and the inequality $|a_1\vee b_1-a_2 \vee b_2|\leqslant |a_1-a_2|\vee|b_1-b_2|$, we have
\begin{equation*}
\begin{split}
|u^\Pi_{i-1}(x_1)-u^\Pi_{i-1}(x_2)| & \leqslant |\tilde{Y}^1_{t_{i-1}}-\tilde{Y}^2_{t_{i-1}}|\vee |g(x_1)-g(x_2)|\\
& \leqslant [(L_i e^{K_1 \Delta t_i }+K_2 \Delta t_i)\vee L_n]|x_1-x_2|.
\end{split}
\end{equation*}
Therefore, we have proved that $u^\Pi_{i-1}$ is Lipschitz continuous and the Lipschitz constant satisfies $L_{i-1}\leqslant (L_i e^{K_1 \Delta t_i }+K_2 \Delta t_i)\vee L_n$, for $i=2,\cdots,n$.
Now it suffices to show that $(L_i)_{1\leqslant i \leqslant n}$ are uniformly bounded. Noting that $L_{i-1}\vee L_n \leqslant (L_i \vee L_n)e^{K_1 \Delta t_i }+K_2 \Delta t_i$, one can apply Lemma \ref{Gronwall} directly to obtain
$$\max_{1\leqslant i\leqslant n}L_i
\leqslant \max_{1\leqslant i\leqslant n}L_i\vee L_n
\leqslant e^{K_1T} (L_n+K_2 T)
=e^{K_1T} (K_g+K_2 T).$$
\end{proof}

At last, we can use the above subsidiary lemmas to obtain the boundness of $\hat{Z}^\Pi$ appearing in the discretization form.

\begin{lemma}\label{bddhatZ}
Suppose Assumption \ref{functions} holds. Then, we have $\hat{Z}^\Pi$ is bounded on $[0,T]$, uniformly in $\Pi$.
\end{lemma}

\begin{proof}
Applying Lemma \ref{bddZ} to the auxiliary forward-backward SDE system (\ref{u-FBSDE}), we can get $|\tilde{Z}_t|\leqslant \exp(KT)[M_\sigma \exp(2K_b T) L_i+1]$ on $[t_{i-1},t_i]$, where $K$ is defined the same as in the previous lemma and the bound is independent of the initial value of the system (\ref{u-FBSDE}). In turn, reviewing Lemma \ref{u_i}, we could obtain that $(\tilde{Y},\tilde{Z})$ and $(\bar{Y}^\Pi,\hat{Z}^\Pi)$ coincide on $[t_{i-1},t_i)$ by setting $x=X_{t_{i-1}}$ in  (\ref{u-FBSDE}), which indicates that $|\hat{Z}^\Pi_t|\leqslant \exp(KT)[M_\sigma \exp(2K_b T) L_i+1]$ on $[t_{i-1},t_i)$. Then by Lemma \ref{uLip}, the uniform boundness of $L_i$ guarantees that  $\hat{Z}^\Pi$ is bounded on the whole $[0,T]$ and the bound does not rely on $\Pi$.
\end{proof}

\subsection{Main result}
Now we are ready to give the main result of this paper.
\begin{theorem}\label{main} 
Let Assumption \ref{functions} hold. Then, we have the following estimate with $q={\frac{1}{2}}$:
\begin{equation}\label{main1}
\sup_{t\in[0,T]}\EE\left[ |\bar{Y}^\Pi_t-Y_t|^2\right]+\sup_{t\in[0,T]}\EE\left[ |\hat{Y}^\Pi_t-Y_t|^2\right] +\EE\left[\int^T_0 |\hat{Z}^\Pi_t-Z_t|^2dt\right]\leqslant C |\Pi|^q,
\end{equation}
\begin{equation}\label{main2}
\max_{1\leqslant i\leqslant n}\EE\left[\sup_{t\in[t_{i-1},t_i]}|\hat{Y}^\Pi_t-Y_t|^2+\sup_{t\in[t_{i-1},t_i]}|\bar{Y}^\Pi_t-Y_t|^2\right] \leqslant C |\Pi|^q
\end{equation}
and
\begin{equation}\label{main3}
\sup_{t\in[0,T]}\EE\left[ |\hat{K}^\Pi_t-K_t|\right]
+\max_{1\leqslant i\leqslant n}\EE\left[\sup_{t\in[t_{i-1},t_i]}|\hat{K}^\Pi_t-K_t|\right]
\leqslant C |\Pi|^{\frac{q}{2}}.
\end{equation}

In addition, if we further assume that $g$ is $C^2_b$, which means it is twice differentiable and all derivatives are uniformly bounded, we can obtain all the above estimates with $q=1$.
\end{theorem}

\begin{proof}
The whole proof is divided into three steps.\\
\noindent\textbf{Step 1.}
Firstly, we claim the following estimate
\begin{equation}\label{step1}
\max_{1\leqslant i\leqslant n}\EE|\bar{Y}^\Pi_{t_{i}}-Y_{t_{i}}|^2+\EE\left[\int^T_0|\hat{Z}_t^\Pi-Z_t|^2dt\right]
\leqslant C |\Pi|^q.
\end{equation}

Recall the discretization form (\ref{Ybar}) and the reflected forward-backward SDE (\ref{M-FBSDE}), and notice that they are based on the same forward SDE. Denote $\Delta Y=Y-\bar{Y}^\Pi$, $\Delta \hat{Y}=Y-\hat{Y}^\Pi$ and $\Delta Z=Z-\hat{Z}^\Pi$. Apply It\^{o}'s formula to $\psi(\Delta Y_t)$ for an increasing $C^2$ function $\psi$ yet to be determined later, and we have for $t\in[t_{i-1},t_i)$,
\begin{equation}\label{psi}
\begin{split}
\psi(\Delta Y_t)=& \psi(\Delta \hat{Y}_{t_i})
+\int^{t_i}_t  \psi^{\prime}(\Delta Y_s)(f(s,X_s,Z_s)-f(s,X_s,\hat{Z}^\Pi_s))ds
-\int^{t_i}_t   \psi^{\prime}(\Delta Y_s)\Delta Z_s dB_s\\
&+\int^{t_i}_t  \psi^{\prime}(\Delta Y_s)dK_s
-\frac{1}{2}\int^{t_i}_t   \psi^{\prime\prime}(\Delta Y_s)|\Delta Z_s|^2ds.
\end{split}
\end{equation}

We deduce from Lemma \ref{bddhatZ} and Assumption \ref{functions} that
\begin{equation}\label{Deltaf}
\left|f(s,X_s,Z_s)-f(s,X_s,\hat{Z}^\Pi_s)\right|\leqslant K_z(1+|Z_s|+|\hat{Z}^\Pi_s|)|\Delta Z_s|
\leqslant K_z(1+2M_z)|\Delta Z_s|+K_z|\Delta Z_s|^2,
\end{equation}
where $M_z$ denotes the uniform bound of $\hat{Z}^\Pi$. Plugging the last inequality into (\ref{psi}) and using the assumption that $\psi$ is increasing, we have from Lemma \ref{bddY} that $|\Delta \hat{Y}_t|\leqslant |\Delta Y_t|$, and

\begin{equation}\label{psi1}
\begin{split}
\psi(\Delta Y_t)\leqslant& \psi(\Delta Y_{t_i}) 
-\int^{t_i}_t  \psi^{\prime}(\Delta Y_s)\Delta Z_s dB_s
+\int^{t_i}_t K_z(1+2M_z) \psi^{\prime}(\Delta Y_s)|\Delta Z_s|ds\\
&+\int^{t_i}_t  \left[K_z\psi^{\prime}(\Delta Y_s)-\frac{1}{2} \psi^{\prime\prime}(\Delta Y_s)\right]|\Delta Z_s|^2ds
+\int^{t_i}_t  \psi^{\prime}(\Delta Y_s)dK_s\\
\leqslant & \psi(\Delta Y_{t_i})
-\int^{t_i}_t   \psi^{\prime}(\Delta Y_s)\Delta Z_s dB_s
+\int^{t_i}_t \frac{K_z}{2}(1+2M_z)^2|\psi^{\prime}(\Delta Y_s)|^2ds\\
&+\int^{t_i}_t  \left[K_z\psi^{\prime}(\Delta Y_s)+\frac{K_z}{2}-\frac{1}{2} \psi^{\prime\prime}(\Delta Y_s)\right]|\Delta Z_s|^2ds
+\int^{t_i}_t  \psi^{\prime}(\Delta Y_s)dK_s,\\
\end{split}
\end{equation}
where the last inequality comes from H\"{o}lder's Inequality.

We now choose $\psi$ with the following form
\begin{equation*}
\psi(x)=\frac{1}{2K_z}(e^{2K_z x}-2K_z x-1),
\end{equation*}
such that $K_z\psi^{\prime}+K_z-\frac{1}{2} \psi^{\prime\prime}=0$, and it is straightforward to check that $\psi$ is a $C^\infty$ function, increasing on $[0,\infty)$ and satisfies $\psi(0)=0$. Furthermore, recalling Lemma \ref{bddY} and the boundness of $Y$ as the solution to forward-backward SDE (\ref{M-FBSDE}) with reflection and denoting $M_\infty\triangleq\Vert\bar{Y}^\Pi\Vert_\infty+\Vert Y\Vert_\infty$, we can then get the following properties of $\psi$ on $[0, M_\infty]$:
\begin{equation}\label{abc}
\begin{split}
&(a)\quad|\psi^{\prime}(x)|^2\leqslant C_1\psi(x),\\
&(b)\quad K_z|x|^2\leqslant \psi(x),\\
&(c)\quad \psi^{\prime}(x)\leqslant C_2 x,
\end{split}
\end{equation}
where $C_1=4K_z e^{2K_z M_\infty}$ and $C_2=C_1/2$.

Set $\tilde{C}\triangleq \frac{K_z}{2}(1+2M_z)^2 C_1$ and $\Lambda_t\triangleq e^{\tilde{C}t}$. Applying It\^{o}'s formula again to $\Lambda_t\psi(\Delta Y_t)$ and noting that $\Delta Y_t \geqslant 0$ by Lemma \ref{bddY}, we have
\begin{equation*}
\begin{split}
&\Lambda_t\psi(\Delta Y_t)+\frac{K_z}{2}\int^{t_i}_t \Lambda_s|\Delta Z_s|^2ds
\leqslant  \Lambda_{t_i}\psi(\Delta Y_{t_i})-\int^{t_i}_t \tilde{C}\Lambda_s\psi(\Delta Y_s)ds\\
&+ \frac{K_z}{2}(1+2M_z)^2 \int^{t_i}_t \Lambda_s|\psi^{\prime}(\Delta Y_s)|^2ds
-\int^{t_i}_t  \Lambda_s\psi^{\prime}(\Delta Y_s)\Delta Z_s dB_s
+\int^{t_i}_t  \Lambda_s \psi^{\prime}(\Delta Y_s)dK_s.\\
\end{split}
\end{equation*}
Noting (\ref{abc})-(a), we further have
\begin{equation}\label{Lpsi}
\begin{split}
\Lambda_t\psi(\Delta Y_t)+\frac{K_z}{2}\int^{t_i}_t \Lambda_s|\Delta Z_s|^2ds
\leqslant & \Lambda_{t_i}\psi(\Delta Y_{t_i})
-\int^{t_i}_t  \Lambda_s\psi^{\prime}(\Delta Y_s)\Delta Z_s dB_s
+\int^{t_i}_t  \Lambda_s \psi^{\prime}(\Delta Y_s)dK_s.\\
\end{split}
\end{equation}

In view of (\ref{abc})-(c), the integrand of the last term in (\ref{Lpsi}) is estimated as follows:
\begin{equation}\label{integrand}
\Lambda_s \psi^{\prime}(\Delta Y_s)dK_s\leqslant C_2 \Lambda_s\Delta Y_s dK_s=C_2 \Lambda_s(Y_s-\bar{Y}^\Pi_s)dK_s, \quad\forall s\in [t_{i-1},t_i].
\end{equation}
Then, the flat-off condition in (\ref{M-FBSDE}), (\ref{Ybar}) and the definition of $\hat{Y}^\Pi$ further yield that
\begin{equation}\label{dK}
\begin{split}
(Y_s-\bar{Y}^\Pi_s)dK_s
=&(g(X_s)-\bar{Y}^\Pi_s)dK_s
=\left[g(X_s)-\EE^{\mathcal{F}_s}\left(\hat{Y}^{\Pi}_{t_i}+\int^{t_i}_s f(r,X_r,\hat{Z}^{\Pi}_r)dr\right)\right]dK_s\\
=&\EE^{\mathcal{F}_s}\left[g(X_s)-\hat{Y}^{\Pi}_{t_i}-\int^{t_i}_s f(r,X_r,\hat{Z}^{\Pi}_r)dr\right]dK_s \\
\leqslant & \EE^{\mathcal{F}_s}\left[g(X_s)-g(X_{t_i})-\int^{t_i}_s f(r,X_r,\hat{Z}^{\Pi}_r)dr\right]dK_s.
\end{split}
\end{equation}

Next, we will consider two cases respectively in order to get finer convergence result when we have additional regularity assumption about the obstacle function $g$. Let Assumption \ref{functions} hold in both cases. We will utilize some standard estimates of forward SDE and use C to denote a universal constant that only depends on $K_g, K_b, M_b$ and $M_\sigma$ at this stage.

\noindent \textbf{Case \uppercase\expandafter{\romannumeral1}.} If $g$ is Lipschitz, we have 
\begin{equation*}
\EE^{\mathcal{F}_s}\left[g(X_s)-g(X_{t_i})\right]
\leqslant K_g\EE^{\mathcal{F}_s}|X_s-X_{t_i}|\leqslant C |\Pi|^{\frac{1}{2}}(1+|X_s|),
 \forall s\in[t_{i-1},t_i].
\end{equation*}

\noindent\textbf{Case \uppercase\expandafter{\romannumeral2}.} If $g$ is further in $C^2_b$, applying It\^o's formula to $g(X_t)$ gives that
\begin{equation*}
g(X_s)-g(X_{t_i})=-\int^{t_i}_s\left[g^\prime (X_r)b(r,X_r)+\frac{1}{2}g^{\prime\prime}(X_r)|\sigma(r)|^2\right]dr-\int^{t_i}_s g^\prime (X_r)\sigma(r)dB_r.
\end{equation*}
Supposing both $|g^\prime|$ and $|g^{\prime\prime}|$ are bounded by $K_g$ and taking conditional expectation, together with the assumptions of $b$ and $\sigma$, we can obtain that
\begin{equation*}
\begin{split}
\EE^{\mathcal{F}_s}\left[g(X_s)-g(X_{t_i})\right]
&\leqslant C \EE^{\mathcal{F}_s}\left[\int^{t_i}_s  (1+|X_r|)dr\right]
=C \int^{t_i}_s  (1+\EE^{\mathcal{F}_s}|X_r|)dr\\
&\leqslant C (t_i-s)(1+|X_s|)
\leqslant C |\Pi|(1+|X_s|),
\forall s\in[t_{i-1},t_i].
\end{split}
\end{equation*}

Combining the two cases together and letting $q=\frac{1}{2}$ when we have Lipschitz obstacle funtion and $q=1$ when considering $C^2_b$ obstacle with more regularity, we get
\begin{equation*}
\EE^{\mathcal{F}_s}\left[g(X_s)-g(X_{t_i})\right]
\leqslant C |\Pi|^q(1+|X_s|)
\leqslant C |\Pi|^q\left[1+\EE^{\mathcal{F}_s}\left(\sup_{0\leqslant t\leqslant T}|X_t|\right)\right],
\forall s\in[t_{i-1},t_i].
\end{equation*}
Plugging it back to (\ref{dK}) and making use of Lemma \ref{bddhatZ}, we have
\begin{equation}\label{dK1}
\begin{split}
(Y_s-\bar{Y}^\Pi_s)dK_s
&\leqslant \EE^{\mathcal{F}_s}\left[g(X_s)-g(X_{t_i})-\int^{t_i}_s f(r,X_r,\hat{Z}^{\Pi}_r)dr\right]dK_s \\
&\leqslant \left[C |\Pi|^q\left(1+\EE^{\mathcal{F}_s}\left(\sup_{0\leqslant t\leqslant T}|X_t|\right)\right)+\int^{t_i}_s \left(M_f+\frac{\alpha}{2}|\hat{Z}^{\Pi}_r|^2 \right)dr\right]dK_s\\
&\leqslant \left[C |\Pi|^q\left(1+\EE^{\mathcal{F}_s}[\mathcal{X}]\right)+\left(M_f+\frac{\alpha}{2}|M_z|^2 \right)|\Pi|\right]dK_s.
\end{split}
\end{equation}
Here we denote $\mathcal{X}\triangleq \sup_{0\leqslant t\leqslant T}|X_t|$, which is a $\mathcal{F}_T$-measurable and square-integrable random variable. Note that we will let the constant $C$ in the following further depend on $T, M_z, M_\infty,\EE|\mathcal{X}|^2,\EE|K_T|^2$ and all the constants appearing in Assumption \ref{functions}, which may vary from line to line as before. In turn, plugging the above estimate into (\ref{integrand}) and taking expectation, we can obtain that 
\begin{equation*}
\begin{split}
&\EE \left[\int^{t_i}_t  \Lambda_s \psi^{\prime}(\Delta Y_s)dK_s\right]
\leqslant C \EE \left[\int^{t_i}_t \Lambda_s (Y_s-\bar{Y}^\Pi_s)dK_s\right]\\
\leqslant &C |\Pi|^q \Lambda_{t_i} \EE \left[\int^{t_i}_t\left(1+\EE^{\mathcal{F}_s}[\mathcal{X}]\right)dK_s\right]
=C |\Pi|^q\Lambda_{t_i}   \EE[\left(1+\mathcal{X}\right)(K_{t_i}-K_t)].
\end{split}
\end{equation*}

Since $\psi^{\prime}(x)=e^{2K_z x}-1$ is bounded on $[0,M_\infty]$ and $Z\in\mathbb{H}^2([0,T];\rm)$, taking expectation on both sides of (\ref{Lpsi}) gives that for any $ t\in[t_{i-1},t_i]$,
\begin{equation}\label{ELpsi}
\EE\left[\Lambda_t\psi(\Delta Y_t) +\frac{K_z}{2}\int^{t_i}_t \Lambda_s|\Delta Z_s|^2ds\right]
\leqslant   \EE[\Lambda_{t_i}\psi(\Delta Y_{t_i})]
+C |\Pi|^q \Lambda_{t_i}   \EE[\left(1+\mathcal{X}\right)(K_{t_i}-K_t)], 
\end{equation}
which further implies 
\begin{equation*}
\begin{split}
\EE\left[\psi(\Delta Y_{t_{i-1}})\right]
&\leqslant  e^{\tilde{C}\Delta t_i} \{\EE[\psi(\Delta Y_{t_i})]
+C |\Pi|^q\EE[\left(1+\mathcal{X}\right)(K_{t_i}-K_{t_{i-1}})]\}\\
&\leqslant  e^{\tilde{C}\Delta t_i} \EE[\psi(\Delta Y_{t_i})]
+C |\Pi|^q\EE[\left(1+\mathcal{X}\right)(K_{t_i}-K_{t_{i-1}})]
\end{split}
\end{equation*}
by letting $t={t_{i-1}}$ and noting $e^{\tilde{C}\Delta t_i} \leqslant \Lambda_T$. Applying Lemma \ref{Gronwall} again and noticing the fact that $\Delta Y_{t_{n}}=0$, we can further obtain
\begin{equation}\label{ELpsi1}
\begin{split}
&\max_{1\leqslant i\leqslant n}\EE\left[\psi(\Delta Y_{t_{i}})\right]
\leqslant  e^{\tilde{C}T}
\left[C |\Pi|^q\sum^n_{i=1}\EE[\left(1+\mathcal{X}\right)(K_{t_i}-K_{t_{i-1}})]\right]\\
=&C |\Pi|^q\EE[\left(1+\mathcal{X}\right)K_T]
\leqslant C |\Pi|^q [\EE(1+|\mathcal{X}|^2)+\EE|K_T|^2]\leqslant C |\Pi|^q.
\end{split}
\end{equation}
Thus, we can conclude from (\ref{abc})-(b) and Lemma \ref{bddY} that
\begin{equation*}
\max_{1\leqslant i\leqslant n}\EE|\Delta \hat{Y}_{t_{i}}|^2
\leqslant\max_{1\leqslant i\leqslant n}\EE|\Delta Y_{t_{i}}|^2
\leqslant C |\Pi|^q.
\end{equation*}

Setting $t=t_{i-1}$ again in (\ref{ELpsi}) and taking summation from $i=1$ to $n$ on both sides give rise to
\begin{equation*}
\begin{split}
&\EE\left[\int^T_0|\Delta Z_s|^2ds\right]
\leqslant \sum^n_{i=1}\EE\left[ \int^{t_i}_{t_{i-1}}\Lambda_s|\Delta Z_s|^2ds\right]\\
\leqslant & C\left[ \EE[\Lambda_{t_n}\psi(\Delta Y_{t_n})]
+C |\Pi|^q\Lambda_T \sum^n_{i=1}\EE[\left(1+\mathcal{X}\right)(K_{t_i}-K_{t_{i-1}})]\right]
\leqslant C |\Pi|^q,
\end{split}
\end{equation*}
through the same arguments as in (\ref{ELpsi1}) and the fact $\psi(0)=0$. Consequently, (\ref{step1}) follows, and it is easy to check the other part of (\ref{main1}).\\


\noindent\textbf{Step 2.}
Taking supremum over $[t_{i-1},t_i]$ on both sides of  (\ref{Lpsi}), we can observe that
\begin{equation}\label{supE}
\begin{split}
\EE\left[\sup_{t\in[t_{i-1},t_i]}\Lambda_t\psi(\Delta Y_t)\right]
&\leqslant\EE[ \Lambda_{t_i}\psi(\Delta Y_{t_i})]
+\EE\left[\sup_{t\in[t_{i-1},t_i]}\left|\int^{t_i}_t  \Lambda_s\psi^{\prime}(\Delta Y_s)\Delta Z_s dB_s\right|\right]\\
&+\EE\left[\sup_{t\in[t_{i-1},t_i]}\int^{t_i}_t  \Lambda_s \psi^{\prime}(\Delta Y_s)dK_s\right].
\end{split}
\end{equation}

For the second term in the above inequality, applying B-D-G inequality and using (\ref{abc})-(a) and Young's Inequality, we obtain
\begin{equation*}
\begin{split}
&\EE\left[\sup_{t\in[t_{i-1},t_i]}\left|\int^{t_i}_t  \Lambda_s\psi^{\prime}(\Delta Y_s)\Delta Z_s dB_s\right|\right]
\leqslant  C \EE\left[\int^{t_i}_{t_{i-1}} | \Lambda_s\psi^{\prime}(\Delta Y_s)\Delta Z_s|^2 ds\right]^{1/2}\\
\leqslant& C\EE\left[\sup_{t\in[t_{i-1},t_i]} \Lambda_t\psi(\Delta Y_t) \int^{t_i}_{t_{i-1}}\Lambda_s|\Delta Z_s|^2 ds\right]^{1/2}\\
\leqslant & \frac{1}{2} \EE\left[\sup_{t\in[t_{i-1},t_i]} \Lambda_t\psi(\Delta Y_t) \right]
+C \EE\left[\int^{t_i}_{t_{i-1}}\Lambda_s|\Delta Z_s|^2 ds\right].
\end{split}
\end{equation*}
In turn, it follows that the third term is equal to
\begin{equation*}
\begin{split}
\EE\left[\int^{t_i}_{t_{i-1}}  \Lambda_s \psi^{\prime}(\Delta Y_s)dK_s\right]
\leqslant & C|\Pi|^q \Lambda_{t_i} \EE[(1+\mathcal{X})(K_{t_i}-K_{t_{i-1}})]\\
\leqslant& C|\Pi|^q \Lambda_T \EE[(1+\mathcal{X})K_T].
\end{split}
\end{equation*}
Plugging them back into (\ref{supE}) gives
\begin{equation*}
\begin{split}
\EE\left[\sup_{t\in[t_{i-1},t_i]} \Lambda_t\psi(\Delta Y_t) \right]
&\leqslant 2 \EE[ \Lambda_{t_i}\psi(\Delta Y_{t_i})]
+C \EE\left[\int^{t_i}_{t_{i-1}}\Lambda_s|\Delta Z_s|^2 ds\right]
+C|\Pi|^{q} \Lambda_T \EE[(1+\mathcal{X})K_T]\\
&\leqslant 2\Lambda_T \EE[ \psi(\Delta Y_{t_i})]
+C\Lambda_T \EE\left[\int^T_0 |\Delta Z_s|^2 ds\right]
+C|\Pi|^{q} \Lambda_T \EE[(1+\mathcal{X})K_T]\\
&\leqslant C\EE[ \psi(\Delta Y_{t_i})]+C|\Pi|^{q}.
\end{split}
\end{equation*}
Then by the result of the first step, we deduce that
\begin{equation*}
\max_{1\leqslant i\leqslant n} \EE\left[\sup_{t\in[t_{i-1},t_i]} \Lambda_t\psi(\Delta Y_t) \right]
\leqslant C \max_{1\leqslant i\leqslant n} \EE[ \psi(\Delta Y_{t_i})]+C|\Pi|^{q}
\leqslant C|\Pi|^{q},
\end{equation*}
and thus (\ref{main2}) follows.\\


\noindent\textbf{Step 3.}
Now we need to check the assertion related to $K$. From (\ref{M-FBSDE}), (\ref{Ybar}) and (\ref{Yhat}), we have
\begin{equation*}
\begin{split}
&\hat{K}^{\Pi}_{t}=\hat{K}^{\Pi}_{t_i}=\hat{Y}^{\Pi}_0
-[\bar{Y}^{\Pi}_{t}\mathbbm{1}_{\{t\neq t_i\}}+\hat{Y}^{\Pi}_{t}\mathbbm{1}_{\{t= t_i\}}]
-\int^{t}_0 f(r,X_r,\hat{Z}^{\Pi}_r)dr+\int^{t}_0 \hat{Z}^{\Pi}_r dB_r,\\
&K_t=Y_0-Y_t
-\int^{t}_0 f(r,X_r,Z_r)dr+\int^{t}_0 Z_r dB_r.
\end{split}
\end{equation*}
Denote $\Delta K \triangleq K-\hat{K}^{\Pi}$. It then follows that
\begin{equation*}
\Delta K_t=
\Delta \hat{Y}_0-[\Delta Y_t \mathbbm{1}_{\{t\neq t_i\}}+\Delta \hat{Y}_t\mathbbm{1}_{\{t= t_i\}}]
-\int^{t}_0 [f(r,X_r,Z_r)-f(r,X_r,\hat{Z}^{\Pi}_r)]dr+\int^{t}_0 \Delta Z_r dB_r.
\end{equation*}

Applying B-D-G inequality and moment inequality, together with the results proved in the former two steps and the estimate (\ref{Deltaf}), we then deduce that

\begin{equation*}
\begin{split}
\EE\left[\sup_{t\in[t_{i-1},t_i]} |\Delta K_t|\right]
\leqslant &\EE|\Delta \hat{Y}_0|+\EE\left[\sup_{t\in[t_{i-1},t_i]}\left( |\Delta Y_t|+ |\Delta \hat{Y}_t|\right)\right]
+\EE\left[\sup_{t\in[t_{i-1},t_i]} \left|\int^{t}_0 \Delta Z_r dB_r \right|  \right]\\
&+\EE\left[\sup_{t\in[t_{i-1},t_i]} \left|\int^{t}_0 [f(r,X_r,Z_r)-f(r,X_r,\hat{Z}^{\Pi}_r)]dr\right| \right]\\
\leqslant &[ \EE|\Delta \hat{Y}_0|^2]^{\frac{1}{2}}
+\left[  \EE\left(\sup_{t\in[t_{i-1},t_i]}\left( |\Delta Y_t|^2+ |\Delta \hat{Y}_t|^2\right)\right) \right]^{\frac{1}{2}}+\EE \left[ \left(\int^{t_i}_0| \Delta Z_r|^2 dr \right)^{\frac{1}{2}} \right]\\
&+\EE\left[\int^{t_i}_0 K_z(1+2M_z+|\Delta Z_r|) |\Delta Z_r|dr\right]\\
\leqslant & C|\Pi|^{\frac{q}{2}}+
C\left(\EE\left[\int^{t_i}_0 (1+|\Delta Z_r|)^2dr\right]\right)^{\frac{1}{2}} 
\left(\EE\left[\int^{t_i}_0 |\Delta Z_r|^2dr\right]\right)^{\frac{1}{2}} \\
\leqslant & C|\Pi|^{\frac{q}{2}}+
C|\Pi|^{\frac{q}{2}}(T+|\Pi|^{q})^{\frac{1}{2}} 
\leqslant  C|\Pi|^{\frac{q}{2}}.
\end{split}
\end{equation*}
Thus we obtain 
\begin{equation*}
\sup_{t\in[0,T]}\EE\left[ |\hat{K}^\Pi_t-K_t|\right]
\leqslant \max_{1\leqslant i\leqslant n}\EE\left[\sup_{t\in[t_{i-1},t_i]}|\hat{K}^\Pi_t-K_t|\right]
\leqslant C |\Pi|^{\frac{q}{2}},
\end{equation*}
and the proof is complete.
\end{proof}\\

\section{Conclusions}

We have characterized the continuous and discrete optimal investment stopping problem separately and provided the convergence result, which comes down to the convergence from discretely to continuously quadratic reflected BSDE, via the tools of quadratic BSDE with bounded terminals. While at present, we need the bounded assumption due to technical restriction when we try to apply the method in Lipschitz case to Quadratic, and what we discussed here is actually a specific form of quadratic generator without $y$ involved, since the BSDE essentially originates from the utility maximization problem. We may further consider the real discrete scheme for the quadratic reflected BSDE which will indicate the way to solve the optimal investment stopping problem numerically, as well as generalize the settings about generator and terminal value in future research.


\begin{appendix}
\section{Proof of  Lemma \ref{bddZ}}
\begin{proof}
As in the literature, we suppose that the functions $b,g$ and $f$ in the forward-backward SDE (\ref{M-BSDE}) are differentiable with respect to $x$ and $z$ firstly. Thus the solution $(X,Y,Z)$ is differentiable with respect to $x$ and $(\nabla X,\nabla Y,\nabla Z)$ satisfies the following SDE and BSDE
\begin{equation}\label{ASDE&BSDE}
\begin{split}
\nabla X_t &=1+\int^t_0 \nabla b(s,X_s)\nabla X_s ds,\\
\nabla Y_t &=\nabla g(X_T)\nabla X_T-\int^T_t \nabla Z_s dB_s
+\int^T_t  [\nabla_x f(s,X_s,Z_s)\nabla X_s+\nabla_z f(s,X_s,Z_s)\nabla Z_s ]ds.
\end{split}
\end{equation}

Moreover, we can deduce from Assumption \ref{functions} that the coefficients appearing in the above equations satisfy $|\nabla b(t,x)| \leqslant K_b$, $|\nabla g(x)| \leqslant K_g$, $|\nabla_x f(t,x,z)| \leqslant K_x(1+|z|)$ and $|\nabla_z f(t,x,z)| \leqslant K_z(1+2|z|)$ respectively. 

Thanks to the Mallivian calculus, it is classical to show that a version of $(Z_t)_{t\in[0,T]}$ is given by $(\nabla Y_t  (\nabla X_t )^{-1} \sigma(t))_{t\in[0,T]}$. Then, noting that both $|\nabla X_t|$ and $|(\nabla X_t )^{-1} |$ are bounded by $e^{K_b T}$ for any $t \in[0,T]$, we have the following estimate
\begin{equation}\label{Afpartialx}
\begin{split}
&|\nabla_x f(s,X_s,Z_s)\nabla X_s|\leqslant K_x(1+|Z_s|)|\nabla X_s|\\
\leqslant &K_x(1+|\nabla Y_t  (\nabla X_t )^{-1} \sigma(t)|)e^{K_b T}\leqslant K_x e^{K_b T}(1+e^{K_b T} M_\sigma |\nabla Y_t| ).
\end{split}
\end{equation}
Let $K\triangleq K_x  M_\sigma e^{2K_b T}$. Applying It\^{o}-Tanaka's formula to $e^{Kt}|\nabla Y_t |$, we obtain 
\begin{equation}\label{AunderB}
\begin{split}
e^{Kt}|\nabla Y_t | =& e^{KT}|\nabla g(X_T)\nabla X_T|
-\int^T_t K e^{Ks}|\nabla Y_s| ds
-\int^T_t sgn(\nabla Y_s)e^{Ks} \nabla Z_s dB_s\\
&+\int^T_t sgn(\nabla Y_s)e^{Ks} [\nabla_x f(s,X_s,Z_s)\nabla X_s+\nabla_z f(s,X_s,Z_s)\nabla Z_s ]ds\\
&-\int^T_t  e^{Ks}dL_s,
\end{split}
\end{equation}
where $L$ is a real-valued, adapted, increasing and continuous process known as local time of $\nabla Y$ at level $0$. The BMO property of $Z\ast B$ and the fact that $|\nabla_z f(t,x,z)| \leqslant K_z(1+2|z|)$ guarantee that
\begin{equation*}
\begin{split}
&\left|\left|\int^{\cdot}_0 \nabla_z f(s,X_s,Z_s)dB_s\right|\right|^2_{BMO}
=\sup_{\tau \in [0,T]} \EE \left[ \int^{T}_\tau |\nabla_z f(s,X_s,Z_s)|^2ds \Big| \mathcal{F}_\tau \right]\\
\leqslant & C \left( 1+\sup_{\tau \in [0,T]} \EE \left[ \int^{T}_\tau |Z_s|^2ds \Big| \mathcal{F}_\tau \right]\right)
=C(1+\left|\left|Z\ast B\right|\right|^2_{BMO}) < \infty,
\end{split}
\end{equation*}
which further implies that $\mathcal{E}(\int^{\cdot}_0 \nabla_z f(s,X_s,Z_s)dB_s)_t$ is a uniformly integrable martingale. In turn, we are able to apply Girsanov theorem and rewrite (\ref{AunderB}) under the equivalent probability $\mathbb{Q}$ as
\begin{equation*}
\begin{split}
e^{Kt}|\nabla Y_t | 
\leqslant &  e^{KT}|\nabla g(X_T)\nabla X_T|
-\int^T_t K e^{Ks}|\nabla Y_s| ds
-\int^T_t sgn(\nabla Y_s)e^{Ks} \nabla Z_s dB^\mathbb{Q}_s\\
&+ \int^T_t e^{Ks} (K_x e^{K_b T}+K|\nabla Y_s | )ds\\
\leqslant &  e^{KT}|\nabla g(X_T)\nabla X_T|
-\int^T_t sgn(\nabla Y_s)e^{Ks} \nabla Z_s dB^\mathbb{Q}_s
+\frac{1}{K}K_x e^{K_b T} e^{KT},\\
\end{split}
\end{equation*}
where we used the estimate (\ref{Afpartialx}) and the fact $dL_t\geqslant 0$, and $B^\mathbb{Q}_t\triangleq B_t-\int^t_0  \nabla_z f(s,X_s,Z_s)ds$ is a standard Brownian motion under $\mathbb{Q}$.
Then, taking conditional expectation on both sides and noticing that $\nabla Z$ is actually the second component of the solution to BSDE in (\ref{ASDE&BSDE}), we obtain
\begin{equation*}
\begin{split}
e^{Kt}|\nabla Y_t | 
\leqslant& \EE^\mathbb{Q}\left[  e^{KT}|\nabla g(X_T)\nabla X_T|
+\frac{1}{K}K_x e^{K_b T} e^{KT} \Big| \mathcal{F}_t \right]\\
\leqslant&  e^{KT}e^{K_b T}[K_g+K_x/K].
\end{split}
\end{equation*}
Using the expression of $Z_t$ again, we can finally deduce that for any $t\in[0,T]$,
\begin{equation*}
\begin{split}
&|Z_t|= |\nabla Y_t  (\nabla X_t )^{-1} \sigma(t)|\\
\leqslant & e^{K_b T}M_\sigma|\nabla Y_t | 
\leqslant e^{KT}[e^{2K_b T}M_\sigma K_g+1].
\end{split}
\end{equation*}

We conclude the proof by noting that when $b,g$ and $f$ are not differentiable, one can also prove the result by a standard approximation and stability results for BSDEs.
\end{proof}
\end{appendix}



\end{document}